\documentclass{llncs}
\usepackage{tikz} 
\usepackage{color}
\usepackage{units,amsmath,amsfonts,amssymb,xspace,stmaryrd}

\newcommand{\ket}[1]{\left\vert{#1}\right\rangle}

\newcommand{\CZ}{\ensuremath{\wedge Z}\xspace}

\newcommand{\fdhilb}{\bf FdHilb}
\newcommand{\wfdhilb}{\fdhilb_{\text{wp}}}
\newcommand{\catD}{\ensuremath{{\cal D}}\xspace}
\newcommand{\id}[1]{\ensuremath{1_{#1}}}
\newcommand{\denote}[1]{\llbracket #1 \rrbracket}

\newcommand{\suck}{\vspace{-1mm}}
\newcommand{\suckb}{\vspace{-1.5mm}}
\newcommand{\suckbb}{\vspace{-2mm}}
\newcommand{\suckbbb}{\vspace{-3mm}}
\newcommand{\suckbbbb}{\vspace{-4mm}}

\spnewtheorem*{spiderlaw}{Spider Law}{\bfseries}{\itshape}

\begin{document}

\title{Graphs States and the necessity of Euler Decomposition}
\author{Ross Duncan\inst{1} \and Simon Perdrix\inst{2,3}} 
\institute{Oxford University Computing Laboratory\\
    Wolfson Building, Parks Road, OX1 3QD Oxford, UK\\
    \email{ross.duncan@comlab.ox.ac.uk}
\and
LFCS, University of Edinburgh, UK\\
\and
PPS, Universit\'e Paris Diderot, France\\
\email{simon.perdrix@pps.jussieu.fr}
}

\maketitle

\begin{abstract}
Coecke and Duncan recently introduced a categorical formalisation of the 
interaction of complementary quantum observables. In this paper we use 
their diagrammatic language to study graph states, a computationally
interesting class of quantum 
states.  We give a graphical proof of
the fixpoint property of graph states. We then introduce a new
equation, for the  Euler decomposition of the Hadamard gate, and  
demonstrate that Van den Nest's theorem---locally
equivalent graphs represent the same entanglement---is equivalent to
this new axiom.   Finally we prove that the Euler decomposition equation
is not derivable from the existing axioms.
\end{abstract}


\section{Introduction}

Upon asking the question ``What are the axioms of quantum mechanics?''
we can expect to hear the usual story about states being vectors of
some Hilbert space, evolution in time being determined by unitary
transformations, etc.  
However, even before finishing chapter one of
the textbook, we surely notice that something is amiss.  Issues around
normalisation, global phases, etc. point to an ``impedence mismatch''
between the theory of quantum 
mechanics and the mathematics used to formalise it.  The question
therefore should be ``What are the axioms of quantum mechanics
\emph{without Hilbert spaces'?}''


In their seminal paper \cite{AbrCoe:CatSemQuant:2004} Abramsky and
Coecke approached this question by studying the categorical structures 
necessary to carry out certain quantum information processing tasks.
The categorical treatment provides as an intuitive pictorial formalism
where quantum states and processes are represented as certain
diagrams, and equations between them are described by rewriting
diagrams.  A recent contribution to this programme was Coecke and
Duncan's axiomatisation of the algebra of a pair complementary
observables \cite{Coecke:2008nx} in terms 
of the \emph{red-green calculus}.  The formalism, while quite
powerful, is known to be incomplete in the following sense: there exist
true equations which are not derivable from the axioms.

In this paper we take one step towards its completion.  We use the
red-green language to study \emph{graph states}.    Graph states \cite{HDERNB06-survey} are
very important class of states used in quantum information processing,
in particular with relation to the  one-way model of quantum
computing \cite{Raussendorf-2001}.   Using the
axioms of the red-green system, we attempt to prove Van den Nest's
theorem \cite{VdN04}, which establishes the local complementation
property for graph states.  In so doing we show that a new equation
must be added  to the system, namely that expressing the  Euler
decomposition of the Hadamard gate.  More precisely, we show that Van
den Nest's theorem is equivalent to the decomposition of $H$, and that
this equation cannot be deduced from the existing axioms of the system.

The paper procedes as follows:  we introduce the graphical language
and the axioms of the red-green calculus, and its basic properties;  we
then introduce graph states and and prove the fixpoint property
of graph states within the calculus; we state Van den
Nest's theorem, and prove our main result---namely that the theorem is
equivalent to the Euler decomposition of $H$.  Finally we demonstrate
a model of the red-green axioms where the Euler decomposition does not
hold, and conclude that this is indeed a new axiom which should be
added to the system.

\section{The Graphical Formalism}
\label{sec:graphical-formalism}

\begin{definition}
  A \emph{diagram} is a finite undirected open graph generated by the
  following two families of vertices:
  \begin{gather*}
    \delta_Z =  \vcenter
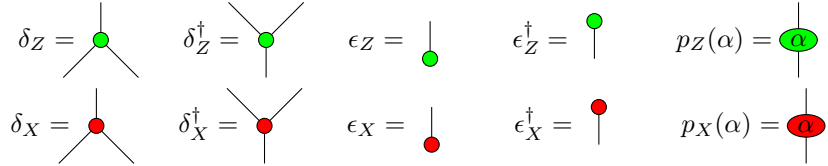
{\hbox{\begin{tikzpicture}[xscale=0.50,yscale=0.50]
\useasboundingbox (-0.5,-0.5) rectangle (0.5,0.5);
\draw (0.00,0.00) -- (-1.00,-1.00);
\draw (0.00,0.00) -- (1.00,-1.00);
\draw (0.00,1.00) -- (0.00,0.00);
\filldraw[fill=green] (0.00,0.00) ellipse (0.20cm and 0.20cm);
\end{tikzpicture}
}}  \qquad  \delta^\dag_Z =  \vcenter{\hbox{\begin{tikzpicture}[xscale=0.50,yscale=0.50]
\useasboundingbox (-0.5,-0.5) rectangle (0.5,0.5);
\draw (0.00,0.00) -- (-1.00,1.00);
\draw (0.00,0.00) -- (0.00,-1.00);
\draw (1.00,1.00) -- (0.00,0.00);
\filldraw[fill=green] (0.00,0.00) ellipse (0.20cm and 0.20cm);
\end{tikzpicture}
}}  \qquad 
    \epsilon_Z = \vcenter{\hbox{\begin{tikzpicture}[xscale=0.50,yscale=0.50]
\useasboundingbox (-0.5,-0.5) rectangle (0.5,1.5);
\draw (0.00,0.00) -- (0.00,1.00);
\filldraw[fill=green] (0.00,0.00) ellipse (0.20cm and 0.20cm);
\end{tikzpicture}
}}  \qquad \epsilon^\dag_Z = \vcenter{\hbox{\begin{tikzpicture}[xscale=0.50,yscale=0.50]
\useasboundingbox (-0.5,-0.5) rectangle (0.5,1.5);
\draw (0.00,1.00) -- (0.00,0.00);
\filldraw[fill=green] (0.00,1.00) ellipse (0.20cm and 0.20cm);
\end{tikzpicture}
}}  \qquad p_Z(\alpha) = \vcenter{\hbox{\begin{tikzpicture}[xscale=0.50,yscale=0.50]
\useasboundingbox (-0.5,-0.5) rectangle (0.5,0.5);
\draw (0.00,0.00) -- (0.00,1.00);
\draw (0.00,-1.00) -- (0.00,0.00);
\filldraw[fill=green] (0.00,0.00) ellipse (0.50cm and 0.30cm);
\draw (0.00,0.00) node{$\alpha$};
\end{tikzpicture}
}} \\
    \delta_X =  \vcenter{\hbox{\begin{tikzpicture}[xscale=0.50,yscale=0.50]
\useasboundingbox (-0.5,-0.5) rectangle (0.5,0.5);
\draw (0.00,0.00) -- (-1.00,-1.00);
\draw (0.00,0.00) -- (1.00,-1.00);
\draw (0.00,1.00) -- (0.00,0.00);
\filldraw[fill=red] (0.00,0.00) ellipse (0.20cm and 0.20cm);
\end{tikzpicture}
}}  \qquad  \delta^\dag_X =  \vcenter{\hbox{\begin{tikzpicture}[xscale=0.50,yscale=0.50]
\useasboundingbox (-0.5,-0.5) rectangle (0.5,0.5);
\draw (0.00,0.00) -- (-1.00,1.00);
\draw (0.00,0.00) -- (0.00,-1.00);
\draw (1.00,1.00) -- (0.00,0.00);
\filldraw[fill=red] (0.00,0.00) ellipse (0.20cm and 0.20cm);
\end{tikzpicture}
}}  \qquad 
    \epsilon_X = \vcenter{\hbox{\begin{tikzpicture}[xscale=0.50,yscale=0.50]
\useasboundingbox (-0.5,-0.5) rectangle (0.5,1.5);
\draw (0.00,0.00) -- (0.00,1.00);
\filldraw[fill=red] (0.00,0.00) ellipse (0.20cm and 0.20cm);
\end{tikzpicture}
}}  \qquad \epsilon^\dag_X = \vcenter{\hbox{\begin{tikzpicture}[xscale=0.50,yscale=0.50]
\useasboundingbox (-0.5,-0.5) rectangle (0.5,1.5);
\draw (0.00,1.00) -- (0.00,0.00);
\filldraw[fill=red] (0.00,1.00) ellipse (0.20cm and 0.20cm);
\end{tikzpicture}
}}  \qquad p_X(\alpha) =\vcenter{\hbox{\begin{tikzpicture}[xscale=0.50,yscale=0.50]
\useasboundingbox (-0.5,-0.5) rectangle (0.5,0.5);
\draw (0.00,0.00) -- (0.00,1.00);
\draw (0.00,-1.00) -- (0.00,0.00);
\filldraw[fill=red] (0.00,0.00) ellipse (0.50cm and 0.30cm);
\draw (0.00,0.00) node{$\alpha$};
\end{tikzpicture}
}} 
  \end{gather*}
  where $\alpha \in [0,2\pi)$,  and a vertex $H = \vcenter{\hbox{\begin{tikzpicture}[xscale=0.50,yscale=0.50]
\useasboundingbox (-0.5,-0.5) rectangle (0.5,2.5);
\draw (0.00,2.00) -- (0.00,1.00) -- (0.00,0.00);
\filldraw[fill=white] (-0.40, 0.60)  -- (-0.40,1.40) -- (0.40,1.40) -- (0.40,0.60) -- (-0.40,0.60);
\draw (0.00,1.00) node{\tiny $H$};
\end{tikzpicture}
}}$ belonging to
  neither family.
\end{definition}

Diagrams form a monoidal category \catD in the evident way: composition is
connecting up the edges, while tensor is simply putting two diagrams
side by side.  In fact, diagrams form a $\dag$-compact category
\cite{KelLap:comcl:1980,AbrCoe:CatSemQuant:2004} but we will suppress
the details of this and let the pictures speak for themselves.  We
rely here on general results 
\cite{JS:1991:GeoTenCal1,Selinger:dagger:2005,Duncan:thesis:2006} 
which state that a pair diagrams are equal by the axioms of
$\dag$-compact categories exactly when they may be deformed to each
other.

Each family forms a  \emph{basis structure}  \cite{PavlovicD:MSCS08}
with an associated \emph{local phase shift}.  The axioms describing
this structure can be subsumed by the following law.  Define $
\delta_0 = \epsilon^\dag$, $\delta_1 = \id{}$ and $\delta_n = 
(\delta_{n-1} \otimes \id{}) \circ \delta$,
and define $\delta^\dag_n$ similarly. 

\begin{spiderlaw}
  Let $f$ be a connected diagram, with $n$ inputs and $m$ outputs, and
  whose vertices are drawn entirely from one family; then 
  \[
  f = \delta_m \circ p(\alpha) \circ \delta^\dag_n 
  \qquad\qquad \emph{where }
  \alpha = \sum_{p(\alpha_i) \in f} \alpha_i\! \mod 2\pi
  \]
  with the convention that $p(0) = \id{}$.\\
  \suckbbbb\suckb
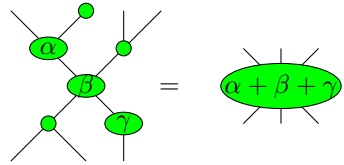
  \begin{center}
    $\vcenter{\hbox{\input{diag/spider-l.tex}}} =
    \vcenter{\hbox{\input{diag/spiderrot2.tex}}}$
  \end{center}
\end{spiderlaw}
\suckbbb

\noindent 
The spider law justifies the use of ``spiders'' in diagrams:
coloured vertices of arbitrary degree labelled by some angle
$\alpha$.  By convention, we leave the vertex empty if $\alpha = 0$.
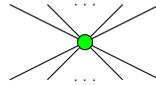
\suckbbb
\[
\vcenter{\hbox{\input{diag/spiderz.tex}}}\suck
\]
We use the spider law as rewrite equation between graphs.  It allows
vertices of the same colour to be merged, or single vertices to be
broken up.  An important special case is when $n = m = 1$ and no angles
occur in $f$; in this case $f$ can be reduced to a simple line.  (This
implies that both families generate the same compact structure.) 
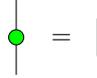
\suckbb
\[
\vcenter{\hbox{\begin{tikzpicture}[xscale=0.50,yscale=0.50]
\useasboundingbox (-0.5,-0.5) rectangle (0.5,0.5);
\draw (0.00,0.00) -- (0.00,-1.00);
\draw (0.00,1.00) -- (0.00,0.00);
\filldraw[fill=green] (0.00,0.00) ellipse (0.20cm and 0.20cm);
\end{tikzpicture}
}} 
= \vcenter{\hbox{\begin{tikzpicture}[xscale=0.50,yscale=0.50]
\useasboundingbox (-0.5,-0.5) rectangle (0.5,1.5);
\draw (0.00,1.00) -- (0.00,0.00) -- (0.00,1.00);
\end{tikzpicture}
}}\suck
\]

\begin{lemma}
  A diagram without $H$ is equal to a bipartite graph.
\end{lemma}
\begin{proof}
  If any two adjacent vertices are the same colour they may be merged
  by the spider law.  Hence if we can do such mergings, every green vertex is
  adjacent only to red vertices, and vice versa.
\end{proof}

\noindent We interpret diagrams in the category  $\wfdhilb$;  this the
category of complex Hilbert spaces and linear maps under the equivalence
relation $f \equiv g$ iff there exists $\theta$ such that $f =
e^{i\theta} g$.  A diagram $f$ with $n$ inputs and $m$ output defines a
linear map  $\denote{f} : \mathbb{C}^{\otimes 2n} \to
\mathbb{C}^{\otimes 2m}$.  Let
\[
\begin{array}{ccc}
  \denote{\epsilon_Z^\dag} = \frac{1}{\sqrt{2}}(\ket{0}+\ket{1}) 
  &  \qquad &
  \denote{\epsilon_X^\dag} = \ket{0} 
  \\ \\
  \denote{\delta_Z^\dag} = \left(
    \begin{array}{cccc} 1&0&0&0\\0&0&0&1 \end{array}
  \right) & \qquad &
  \denote{\delta_X^\dag} =  \frac{1}{\sqrt{2}}\left(
    \begin{array}{cccc} 1&0&0&1\\0&1&1&0 \end{array}
  \right)
  \\ \\
    \denote{p_Z(\alpha)} = \left(
    \begin{array}{cc}
      1&0\\0&e^{i \alpha }
    \end{array} \right)
  & \qquad &
  \denote{p_X(\alpha)} = e^{-\frac{i\alpha}{2}}\left(
    \begin{array}{cc}
      \cos \frac{\alpha}{2} & i \sin \frac{\alpha}{2} \\
      i \sin \frac{\alpha}{2} & \cos \frac{\alpha}{2}
    \end{array} \right)
\end{array}
\]
\[ 
\denote{H} = \frac{1}{\sqrt{2}} \left(
  \begin{array}{cc}
    1&1\\1&-1
  \end{array} \right)
\]
and set $\denote{f^\dag} = \denote{f}^\dag$.  The map $\denote{\cdot}$
extends in the evident way to a monoidal functor.  

The interpretation of \catD contains a universal set of quantum gates.
Note that $p_Z(\alpha)$ and $p_X(\alpha)$ are the rotations around the
$X$ and $Z$ axes, and in particular when $\alpha = \pi$ they yield
the Pauli $X$ and $Z$ matrices. The \CZ is defined by:
\suckbbb
\[
\CZ =  \vcenter{\hbox{\input{diag/controlZ.tex}}} \suckbbb
\]
The $\delta_X$ and $\delta_Z$ maps \emph{copy} the eigenvectors of the
Pauli $X$ and $Z$; the $\epsilon$ maps \emph{erase} them.  (This is why
such structures are called basis structures).  

Now we introduce the equations\footnote{We have, both above and below,
  made some simplifications to the axioms of \cite{Coecke:2008nx}
  which are specific to the case of qubits.  We also supress scalar factors.}  
which make the $X$ and $Z$ families into \emph{complementary} basis
structures as in \cite{Coecke:2008nx}.  Note that all of the equations
are also satisfied in satisfied in the Hilbert space interpretation. We
present them in one colour only; they also hold with the colours reversed.

\begin{description}
\item[Copying] 
\[
\vcenter{\hbox{\begin{tikzpicture}[xscale=0.50,yscale=0.50]
\useasboundingbox (-0.5,-0.5) rectangle (2.5,2.5);
\draw (1.00,1.00) -- (0.00,0.00);
\draw (1.00,1.00) -- (2.00,0.00);
\draw (1.00,2.00) -- (1.00,1.00);
\filldraw[fill=red] (1.00,2.00) ellipse (0.20cm and 0.20cm);
\filldraw[fill=green] (1.00,1.00) ellipse (0.20cm and 0.20cm);
\end{tikzpicture}
}} =\vcenter{\hbox{\begin{tikzpicture}[xscale=0.50,yscale=0.50]
\useasboundingbox (-0.5,-0.5) rectangle (0.5,1.5);
\draw (0.00,1.00) -- (0.00,0.00);
\filldraw[fill=red] (0.00,1.00) ellipse (0.20cm and 0.20cm);
\end{tikzpicture}
}} \vcenter{\hbox{\begin{tikzpicture}[xscale=0.50,yscale=0.50]
\useasboundingbox (-0.5,-0.5) rectangle (0.5,1.5);
\draw (0.00,1.00) -- (0.00,0.00);
\filldraw[fill=red] (0.00,1.00) ellipse (0.20cm and 0.20cm);
\end{tikzpicture}
}} \text{ and } \vcenter{\hbox{\begin{tikzpicture}[xscale=0.50,yscale=0.50]
\useasboundingbox (-0.5,-0.5) rectangle (2.5,2.5);
\draw (1.00,1.00) -- (0.00,0.00);
\draw (1.00,1.00) -- (2.00,0.00);
\draw (1.00,2.00) -- (1.00,1.00);
\filldraw[fill=green] (1.00,2.00) ellipse (0.20cm and 0.20cm);
\filldraw[fill=red] (1.00,1.00) ellipse (0.20cm and 0.20cm);
\end{tikzpicture}
}} =\vcenter{\hbox{\begin{tikzpicture}[xscale=0.50,yscale=0.50]
\useasboundingbox (-0.5,-0.5) rectangle (0.5,1.5);
\draw (0.00,1.00) -- (0.00,0.00);
\filldraw[fill=green] (0.00,1.00) ellipse (0.20cm and 0.20cm);
\end{tikzpicture}
}} \vcenter{\hbox{\begin{tikzpicture}[xscale=0.50,yscale=0.50]
\useasboundingbox (-0.5,-0.5) rectangle (0.5,1.5);
\draw (0.00,1.00) -- (0.00,0.00);
\filldraw[fill=green] (0.00,1.00) ellipse (0.20cm and 0.20cm);
\end{tikzpicture}
}}
\]
\item[Bialgebra]
\suckbbb
\[
\vcenter{\hbox{\input{diag/bialgebra-l.tex}}} = \vcenter{\hbox{\begin{tikzpicture}[xscale=0.50,yscale=0.50]
\useasboundingbox (-0.5,-0.5) rectangle (2.5,3.5);
\draw (1.00,2.00) -- (0.00,3.00);
\draw (2.00,3.00) -- (1.00,2.00);
\draw (1.00,2.00) -- (1.00,1.00) -- (2.00,0.00);
\draw (1.00,1.00) -- (0.00,0.00);
\filldraw[fill=green] (1.00,2.00) ellipse (0.20cm and 0.20cm);
\filldraw[fill=red] (1.00,1.00) ellipse (0.20cm and 0.20cm);
\end{tikzpicture}
}}
\]
\item[$\pi$-Commutation]
\suckbbbb
\[
\vcenter{\hbox{\input{diag/pi-classical-l.tex}}}  
= \vcenter{\hbox{\begin{tikzpicture}[xscale=0.50,yscale=0.50]
\useasboundingbox (-0.5,-0.5) rectangle (2.5,3.5);
\draw (1.00,1.00) -- (1.00,2.00) -- (1.00,3.00);
\draw (1.00,1.00) -- (2.00,0.00);
\draw (1.00,1.00) -- (0.00,0.00);
\filldraw[fill=green] (1.00,2.00) ellipse (0.50cm and 0.30cm);
\filldraw[fill=red] (1.00,1.00) ellipse (0.20cm and 0.20cm);
\draw (1.00,2.00) node{\tiny $\pi$};
\end{tikzpicture}
}}
~~~~
\vcenter{\hbox{\begin{tikzpicture}[xscale=0.50,yscale=0.50]
\useasboundingbox (-0.5,-0.5) rectangle (1.5,2.5);
\draw (1.00,0.00) -- (1.00,1.00) -- (1.00,2.00);
\filldraw[fill=green] (1.00,1.00) ellipse (0.50cm and 0.30cm);
\filldraw[fill=red] (1.00,0.00) ellipse (0.20cm and 0.20cm);
\draw (1.00,1.00) node{\tiny $\pi$};
\end{tikzpicture}
}} 
= \vcenter{\hbox{\begin{tikzpicture}[xscale=0.50,yscale=0.50]
\useasboundingbox (-0.5,-0.5) rectangle (0.5,1.5);
\draw (0.00,0.00) -- (0.00,1.00);
\filldraw[fill=red] (0.00,0.00) ellipse (0.20cm and 0.20cm);
\end{tikzpicture}
}} 
\suckbbb
\]
\end{description}
\noindent 
A consequence of the axioms we have presented so far is the Hopf Law:
\[
\vcenter{\hbox{\input{diag/hopf-l.tex}}}  = \vcenter{\hbox{\begin{tikzpicture}[xscale=0.50,yscale=0.50]
\useasboundingbox (-0.5,-0.5) rectangle (2.5,2.5);
\draw (1.00,3.00) -- (1.00,2.00);
\draw (1.00,0.00) -- (1.00,-1.00);
\filldraw[fill=green] (1.00,2.00) ellipse (0.20cm and 0.20cm);
\filldraw[fill=red] (1.00,0.00) ellipse (0.20cm and 0.20cm);
\end{tikzpicture}
}}
\]
This equation, when combined with the spider law, provides a very useful
property, namely that every diagram (without $H$) is equal to one
without parallel edges.
\begin{lemma}\label{lem:no-plll-edges}
  Every diagram without $H$ is equal to one without parallel edges.
\end{lemma}
\begin{proof}
  Suppose that $v,u$ are vertices in some diagram, connected by two or
  more edges.  If they are the same colour, they can be joined by the
  spider law, eliminating the edges between them.  Otherwise the Hopf
  law allows one pair of parallel edges to be removed; the result
  follows by induction.
\end{proof}
Finally, we introduce the equations for $H$: \suckbbbb
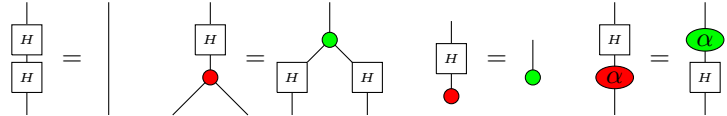
\[
\vcenter{\hbox{\input{diag/h-sq.tex}}}  =  \vcenter{\hbox{\begin{tikzpicture}[xscale=0.50,yscale=0.50]
\useasboundingbox (-0.5,-0.5) rectangle (0.5,3.5);
\draw (0.00,3.00) -- (0.00,0.00);
\end{tikzpicture}
}}~~
\vcenter{\hbox{\input{diag/h-delta-l.tex}}}=\vcenter{\hbox{\input{diag/h-delta-r.tex}}}
\vcenter{\hbox{\begin{tikzpicture}[xscale=0.50,yscale=0.50]
\useasboundingbox (-0.5,-0.5) rectangle (1.5,2.5);
\draw (1.00,1.00) -- (1.00,0.00);
\draw (1.00,2.00) -- (1.00,1.00);
\filldraw[fill=white] (0.60, 0.60)  -- (0.60,1.40) -- (1.40,1.40) -- (1.40,0.60) -- (0.60,0.60);
\filldraw[fill=red] (1.00,0.00) ellipse (0.20cm and 0.20cm);
\draw (1.00,1.00) node{\tiny $H$};
\end{tikzpicture}
}} = \vcenter{\hbox{\begin{tikzpicture}[xscale=0.50,yscale=0.50]
\useasboundingbox (-0.5,-0.5) rectangle (0.5,1.5);
\draw (0.00,0.00) -- (0.00,1.00);
\filldraw[fill=green] (0.00,0.00) ellipse (0.20cm and 0.20cm);
\end{tikzpicture}
}}~~~~
\vcenter{\hbox{\input{diag/h-rot-l.tex}}} =
~\vcenter{\hbox{\input{diag/h-rot-r.tex}}} 
\]
The special role of $H$ in the system is central to our investigation
in this paper.

\section{Generalised Bialgebra Equations}

The bialgebra law is a key equation in the graphical calculus. Notice
that the left hand side of the equation is a 2-colour bipartite
graph which is both a $K_{2,2}$ (i.e a complete bipartite graph) and
a $C_4$ (i.e. a cycle composed of 4 vertices) with alternating
colours. In the following we introduce two generalisations of the
bialgebra equation, one for any  $K_{n,m}$ and another one for any
$C_{2n}$ (even cycle).  

We give graphical proofs for both generalisations; both proofs rely
essentially on the primitive bialgebra equation. 

\begin{lemma}\label{lem:knmp2}  For any $n,m$, ``$K_{n,m} = P_2$'', graphically: 
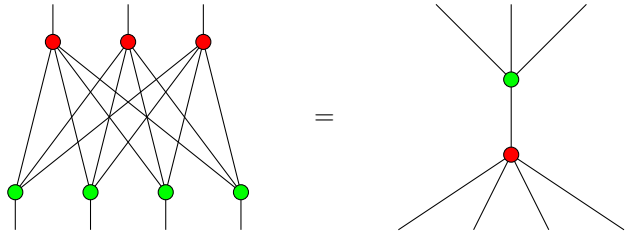
\vspace{-0.6cm}
\[ 
\vcenter{\hbox{\input{diag/k34.tex}}} 
= 
\vcenter{\hbox{\input{diag/e34.tex}}} 
\] 
\end{lemma}
\begin{proof} 
 The proof is by induction on $(m,n)$ where $m$ (resp. $n$) is the
 number of red (resp. green) dots of the left hand side of the
 equation. Let $\prec$ be the lexicographical order (i.e. $(m,n)\prec
 (k,l)$ iff $m<n$ or $m=k \wedge m<l$).  Notice that if either $m =
 1$ or $n = 1$ then the resulting degree 1 vertices may simply be
 removed, by the spider theorem, hence the equation is trivially
 satisfied.  Moreover, if $m=n=2$ the equation is nothing but the
 bialgebra equation.  Let $(m,n)\succ (2,2)$. The following graphical
 proof is by induction, using the hypothesis of induction twice,
 first with $(m,n-1)$ and then with $(m,2)$.  
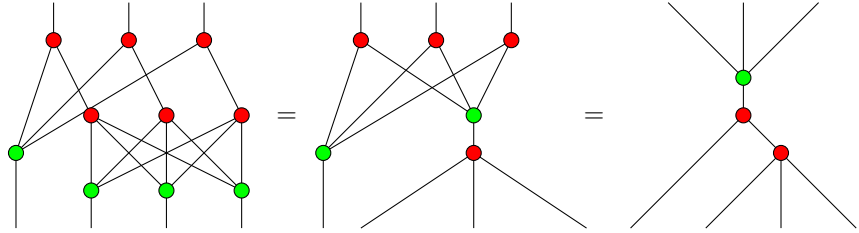
\[ 
\vcenter{\hbox{\input{diag/k34-1.tex}}}
= 
\vcenter{\hbox{\input{diag/k34-2.tex}}}
=
\vcenter{\hbox{\input{diag/k34-3.tex}}}
\] 
Notice in the first step we use the spider law to extract the
$K_{m,n-1}$ subgraph.
\end{proof}

\begin{lemma}\label{lem:Cn}For $n$, an even cycle of size $2n$, of alternating
  colours, can be rewritten into hexagons.  Graphically:
  
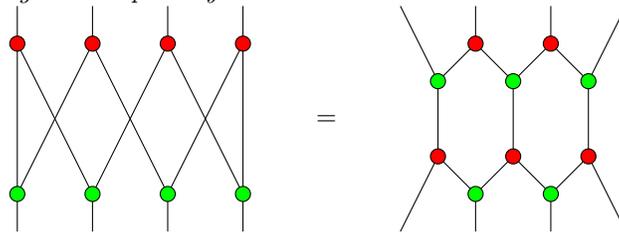
\vspace{-0.6cm}
  \[ 
  \vcenter{\hbox{\input{diag/c8.tex}}} 
  = 
  \vcenter{\hbox{\input{diag/h4.tex}}} 
  \] 
  %
  
  %
\end{lemma}
\begin{proof}
 The proof is by induction, with one application of the bialgebra
 equation:
 \[
 \vcenter{\hbox{\input{diag/c8b.tex}}} =
 \vcenter{\hbox{\input{diag/h4-2b.tex}}} =
 \vcenter{\hbox{\input{diag/h4-1.tex}}}
 \]
 Note, as before, the use of the spider theorem in the first step.
\end{proof}

\section{Graph states}

In order to explore the power and the limits of the axioms we have
described, we now consider the example of \emph{graph states}.  Graph
states provide a good testing ground for our formalism because they are
relatively easy to describe, but have wide applications across quantum
information, for example they form a basis for universal quantum computation,
capture key properties of entanglement, are related to quantum  
error correction, establish links to graph theory and violate Bell
inequalities.

In this section we show how graph states may be defined in the graphical
language, and give a graphical proof of
the fix point property, a fundamental property of graph states.  The
next  section will expose a limitation of the theory, and we will see
that proving Van den Nest's theorem requires an additional axiom.

\begin{definition}
For a given simple undirected graph $G$, let $\ket G$ be the
corresponding graph state  
\[
\ket G = \left(
  \prod_{(u,v)\in E(G)} \CZ_{u,v} 
  \right)
  \left(
    \bigotimes_{u\in V(G)}\frac{\ket 0_u +\ket 1_u}{\sqrt 2}
  \right)
\]
where $V(G)$ (resp. $E(G)$) is the set of vertices (resp. edges) of
$G$. 
\end{definition}
\noindent Notice that for any $u,v,u',v'\in V(G)$, $\CZ_{u,v} = \CZ_{v,u}$ and
$\CZ_{u,v}\circ \CZ_{u',v'} = \CZ_{u',v'}\circ \CZ_{u,v}$, which make
the definition of $\ket G$ does not depends on the orientation or
order of the edges of $G$.

Since both the state $\ket +=\frac{\ket 0+\ket 1}{\sqrt 2}$ and the
unitary gate $\CZ$ can be depicted in the graphical calculus,
any graph state can be represented in the graphical language.   
For instance, the $3$-qubit graph state associated to the triangle is
represented as follows: \suckbb
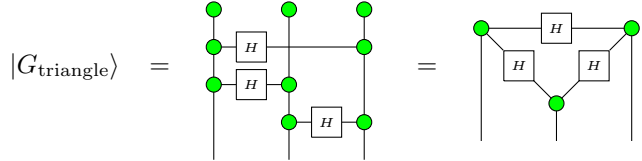
\[ 
\ket {G_\textup{triangle}} 
~~=~~ \vcenter{\hbox{\input{diag/def-tri-1.tex}}} 
~~=~~  \vcenter{\hbox{\input{diag/def-tri-2.tex}}}
\] 

\vspace{0.3cm}

\noindent More generally, any graph $G$, $\ket G$ may be depicted by a
diagram 
composed of $|V(G)|$ green dots. Two green dots are connected with a
$H$ gate if and only if the corresponding vertices are connected in
the graph. Finally, one output wire is connected to every green dot.
Note that the qubits in this picture are the output wires rather than
the dots themselves; to act on a qubit with some operation we simply
connect the picture for that operation to the wire.

Having introduced the graphs states we are now in position to derive
one of their fundamental properties, namely the \emph{fixpoint
property}.

\begin{property}[Fixpoint]\label{lem:fixpoint} Given a graph $G$ and a vertex $u\in V(G)$, 
$$R_x(\pi)^{(u)}R_z(\pi)^{(N_G(u))} \ket G = \ket G$$
\end{property}

The fixpoint property can shown in the graphical calculus by the 
following example.  Consider a star-shaped graph shown 
below; the qubit $u$  is shown at the  top of the diagram, with its
neighbours below.  The fixpoint property simply asserts that the
depicted  equation holds.
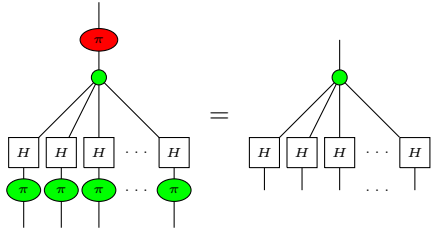
\[ 
\vcenter{\hbox{\input{diag/Fixpoint-l.tex}}}  
= ~
\vcenter{\hbox{\input{diag/Fixpoint-r.tex}}} 
\] 

\begin{theorem}\label{thm:fixpoint}
The fixpoint property is provable in the graphical language.
\end{theorem}
\begin{proof}
 First, notice it is enough to consider star graphs. Indeed, for more
 complicated graphs, green rotations can always be pushed through the
 green dots, leading to the star case.
 
 Let $S_n$ be the star composed of $n$ vertices. Since the red $\pi$-rotation is a green comonoid homorphism, the fixpoint property is satisfied for $S_1$:
 \[ 
\vcenter
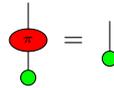
{\hbox{\begin{tikzpicture}[xscale=0.50,yscale=0.50]
\useasboundingbox (-0.5,-0.5) rectangle (0.5,2.5);
\draw (0.00,0.00) -- (0.00,1.00) -- (0.00,2.00);
\filldraw[fill=red] (0.00,1.00) ellipse (0.50cm and 0.30cm);
\filldraw[fill=green] (0.00,0.00) ellipse (0.20cm and 0.20cm);
\draw (0.00,1.00) node{\tiny $\pi$};
\end{tikzpicture}
}} =
\vcenter{\hbox{\begin{tikzpicture}[xscale=0.50,yscale=0.50]
\useasboundingbox (-0.5,-0.5) rectangle (0.5,1.5);
\draw (0.00,0.00) -- (0.00,1.00);
\filldraw[fill=green] (0.00,0.00) ellipse (0.20cm and 0.20cm);
\end{tikzpicture}
}} 
\] 

By induction, for any $n>1$,

 \[ 
\vcenter{\hbox{\input{diag/FP-sn-1.tex}}} =~
\vcenter{\hbox{\input{diag/FP-sn-2.tex}}} =~
\vcenter{\hbox{\input{diag/FP-sn-3.tex}}} =~
\vcenter{\hbox{\input{diag/Fixpoint-r.tex}}} 
\] 
\end{proof}

%

\section{Local Complementation}

In this section, we present the Van den Nest theorem. According to
this theorem, if two graphs are locally equivalent (i.e. one graph can
be transformed into the other by means of local complementations) then
the corresponding quantum states are LC-equivalent, i.e. there
exists a local Clifford unitary\footnote{
One-qubit Clifford unitaries form a finite group generated by $\pi/2$
rotations around $X$ and $Z$ axis: $R_x(\pi/2), R_z(\pi/2)$. A
$n$-qubit local Clifford is the tensor product of $n$ one-qubit
Clifford unitaries.
} which transforms one state into the other. 
 We prove that the local complementation
property is true if and only if $H$ has an Euler decomposition into
$\pi/2$-green and red rotations. 
At the end of the section, we demonstrate that the $\pi/2$
decomposition does not hold in all models of the axioms, and hence
show that the axiom is truly necessary to prove Van den Nest's
Theorem.

\begin{definition}[Local Complementation]
  Given a graph $G$ containing some vertex $u$, we define the
  \emph{local complementation} of $u$ in $G$, written $G*u$ by the
  complementation of the neighbourhood of $u$, i.e.  $V(G*u) = V(G)$,
  $E(G*u):=E(G) \Delta (N_G(u)\times N_G(u))$, where $N_G(u)$ is the
  set of neighbours of $u$ in $G$ ($u$ is not $N_G(u)$) and $\Delta$ is
  the symmetric difference, i.e. $x\in A\Delta B$ iff $x\in A$ xor
  $x\in B$. 
\end{definition}

\begin{theorem}[Van den Nest]\label{vdn}
Given a graph $G$ and a vertex $u\in V(G)$, 
\[
R_x(-\pi/2)^{(u)}
R_z^{(N_G(u))}\ket G=\ket {G*u}\;.
\] 
\end{theorem}

\noindent We illustrate the theorem in the case of a star graph:
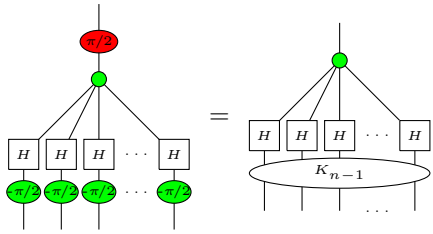
\[ 
\vcenter{\hbox{\input{diag/kn-l.tex}}} 
= ~
\vcenter{\hbox{\input{diag/kn-r.tex}}} 
\]
where $K_{n-1}$ denotes the totally connected graph.

\begin{theorem}\label{thm:LC}
Van den Nest's theorem holds if and only if $H$ can be
decomposed into $\pi/2$ rotations as follows: 
\vspace{-0.7cm}
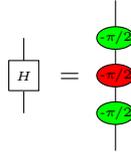
\[  
\vcenter{\hbox{\begin{tikzpicture}[xscale=0.50,yscale=0.50]
\useasboundingbox (-0.5,-0.5) rectangle (0.5,2.5);
\draw (0.00,2.00) -- (0.00,1.00) -- (0.00,0.00);
\filldraw[fill=white] (-0.40, 0.60)  -- (-0.40,1.40) -- (0.40,1.40) -- (0.40,0.60) -- (-0.40,0.60);
\draw (0.00,1.00) node{\tiny $H$};
\end{tikzpicture}
}} 
= ~
\vcenter{\hbox{\input{diag/axiomEuler-r.tex}}} 
\] 
\end{theorem}

\noindent Notice that this equation is nothing but the  Euler decomposition of
$H$: $$H  =R_Z(-\pi/2)\circ R_X(-\pi/2)\circ R_Z(-\pi/2)$$ 
Several interesting consequences follow from the decomposition.  We
note two:

%


\begin{lemma}\label{lem:H-euler-non-unique} The $H$-decomposition into
  $\pi/2$ rotations is not unique: \vspace{-0.5cm}
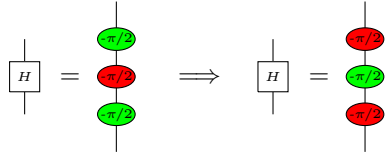
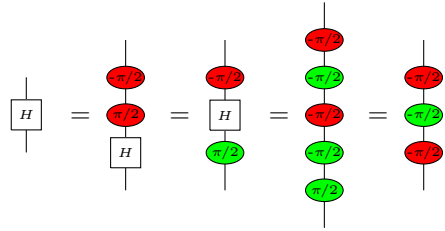
\[ 
\vcenter{\hbox{\begin{tikzpicture}[xscale=0.50,yscale=0.50]
\useasboundingbox (-0.5,-0.5) rectangle (0.5,2.5);
\draw (0.00,2.00) -- (0.00,1.00) -- (0.00,0.00);
\filldraw[fill=white] (-0.40, 0.60)  -- (-0.40,1.40) -- (0.40,1.40) -- (0.40,0.60) -- (-0.40,0.60);
\draw (0.00,1.00) node{\tiny $H$};
\end{tikzpicture}
}} 
= ~
\vcenter{\hbox{\input{diag/axiomEuler-r.tex}}} 
~~\implies~~
\vcenter{\hbox{\begin{tikzpicture}[xscale=0.50,yscale=0.50]
\useasboundingbox (-0.5,-0.5) rectangle (0.5,2.5);
\draw (0.00,2.00) -- (0.00,1.00) -- (0.00,0.00);
\filldraw[fill=white] (-0.40, 0.60)  -- (-0.40,1.40) -- (0.40,1.40) -- (0.40,0.60) -- (-0.40,0.60);
\draw (0.00,1.00) node{\tiny $H$};
\end{tikzpicture}
}} 
= ~
\vcenter{\hbox{\input{diag/axiomEuler-rb.tex}}} 
\] 
\end{lemma}
\begin{proof} 
 \[
 \vcenter{\hbox{\begin{tikzpicture}[xscale=0.50,yscale=0.50]
\useasboundingbox (-0.5,-0.5) rectangle (0.5,2.5);
\draw (0.00,2.00) -- (0.00,1.00) -- (0.00,0.00);
\filldraw[fill=white] (-0.40, 0.60)  -- (-0.40,1.40) -- (0.40,1.40) -- (0.40,0.60) -- (-0.40,0.60);
\draw (0.00,1.00) node{\tiny $H$};
\end{tikzpicture}
}} ~=~
 \vcenter{\hbox{\input{diag/axiomEuler-1.tex}}} ~=~
 \vcenter{\hbox{\input{diag/axiomEuler-2.tex}}} ~=~
 \vcenter{\hbox{\input{diag/axiomEuler-3.tex}}} ~=~
 \vcenter{\hbox{\input{diag/axiomEuler-rb.tex}}}
 \]
\end{proof}

\begin{lemma}\label{lem:pi2-colour-change}  Each colour of $\pi/2$
  rotation may be expressed in terms of the other colour.\vspace{-0.5cm}
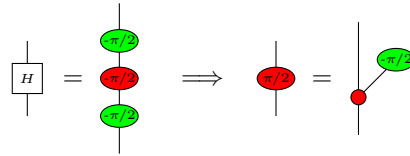
\[ 
\vcenter{\hbox{\begin{tikzpicture}[xscale=0.50,yscale=0.50]
\useasboundingbox (-0.5,-0.5) rectangle (0.5,2.5);
\draw (0.00,2.00) -- (0.00,1.00) -- (0.00,0.00);
\filldraw[fill=white] (-0.40, 0.60)  -- (-0.40,1.40) -- (0.40,1.40) -- (0.40,0.60) -- (-0.40,0.60);
\draw (0.00,1.00) node{\tiny $H$};
\end{tikzpicture}
}} 
= ~
\vcenter{\hbox{\input{diag/axiomEuler-r.tex}}} 
~~\implies~~
\vcenter
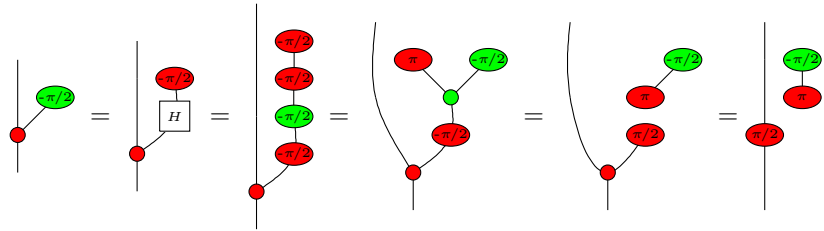
{\hbox{\begin{tikzpicture}[xscale=0.50,yscale=0.50]
\useasboundingbox (-0.5,-0.5) rectangle (0.5,2.5);
\draw (0.00,2.00) -- (0.00,1.00) -- (0.00,0.00);
\filldraw[fill=red] (0.00,1.00) ellipse (0.50cm and 0.30cm);
\draw (0.00,1.00) node{\tiny $\nicefrac{\pi}{2}$};
\end{tikzpicture}
}} 
= 
\vcenter{\hbox{\input{diag/hpi_6.tex}}} 
\] 
\end{lemma}
\begin{proof}
 \[
 \vcenter{\hbox{\input{diag/hpi_6.tex}}} =
 \vcenter{\hbox{\input{diag/hpi-1.tex}}} =
 \vcenter{\hbox{\input{diag/hpi-2.tex}}} =
 \vcenter{\hbox{\input{diag/hpi-3.tex}}} =
 \vcenter{\hbox{\input{diag/hpi-4.tex}}} =
 \vcenter{\hbox{\input{diag/hpi-5.tex}}}
 \]
\end{proof}

\noindent \emph{Remark:} The preceding lemmas depend only on the
\emph{existence} of a decomposition of the form $H = R_z(\alpha) \;R_x
(\beta )\; R_z(\gamma )$.   It is straight forward to generalise these
result based on an arbitrary sequence of rotations, although in the
rest of this paper we stick to the concrete case of $\pi/2$.

Most of the rest of the paper is devoted to proving Theorem \ref{thm:LC}: the
equivalence of Van den Nest's theorem and the Euler form of $H$.  We
begin by proving the easier direction: that the Euler decomposition
implies the local complementation property.

\subsection{Euler Decomposition Implies Local Complementation}
\label{sec:eulur-decomp-impl}

\subsubsection{Triangles}

We begin with the simplest non trivial
examples of local complementation, namely triangles. A local
complementation on one vertex of the triangle removes the opposite
edge.  

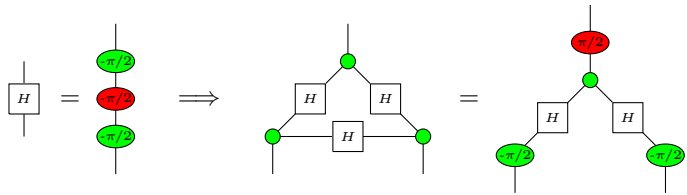
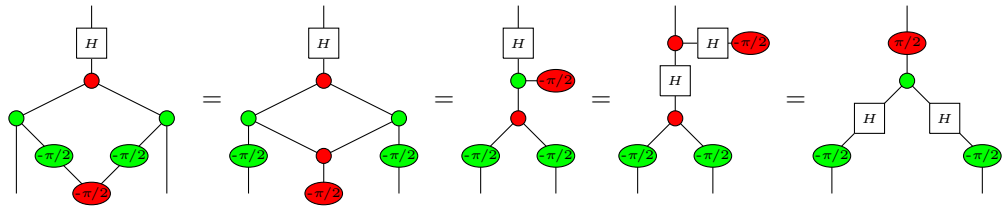
\begin{lemma}\label{lem:triangle} ~
\vspace{-0.8cm}
\[ 
\vcenter{\hbox{\begin{tikzpicture}[xscale=0.50,yscale=0.50]
\useasboundingbox (-0.5,-0.5) rectangle (0.5,2.5);
\draw (0.00,2.00) -- (0.00,1.00) -- (0.00,0.00);
\filldraw[fill=white] (-0.40, 0.60)  -- (-0.40,1.40) -- (0.40,1.40) -- (0.40,0.60) -- (-0.40,0.60);
\draw (0.00,1.00) node{\tiny $H$};
\end{tikzpicture}
}} 
= ~
\vcenter{\hbox{\input{diag/axiomEuler-r.tex}}} 
~~\implies~~
\vcenter{\hbox{\input{diag/tri-l.tex}}} 
= ~
\vcenter{\hbox{\input{diag/tri-r.tex}}} 
\] 
\end{lemma}
\begin{proof}
 \[
 \vcenter{\hbox{\input{diag/tri-1.tex}}} =
 \vcenter{\hbox{\input{diag/tri-2.tex}}} =
 \vcenter{\hbox{\input{diag/tri-3.tex}}} =
 \vcenter{\hbox{\input{diag/tri-4.tex}}} =
 \vcenter{\hbox{\input{diag/tri-r.tex}}}
 \]
 Note the use of Lemma \ref{lem:pi2-colour-change} in the last
 equation.
\end{proof}

\subsubsection{Complete Graphs and Stars}

More generally, $S_n$ (a star composed of $n$ vertices) and
$K_n$ (a complete graph on $n$ vertices) are locally equivalent for all
$n$. 

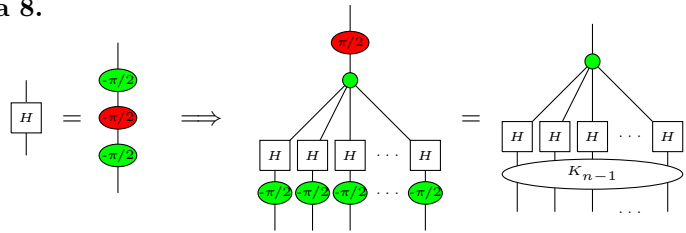
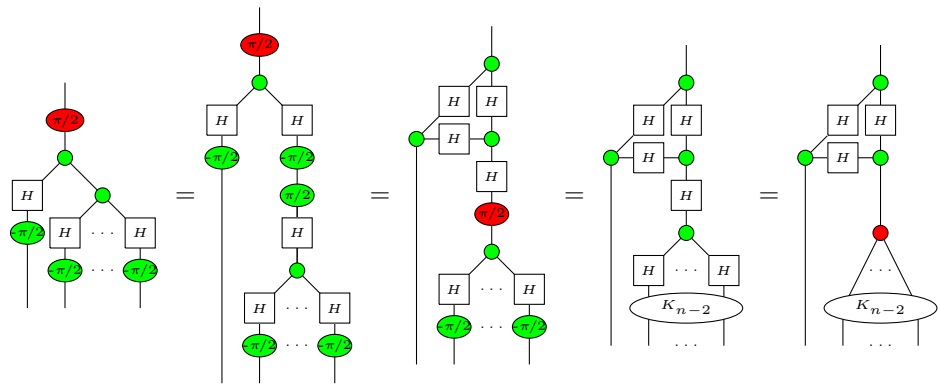
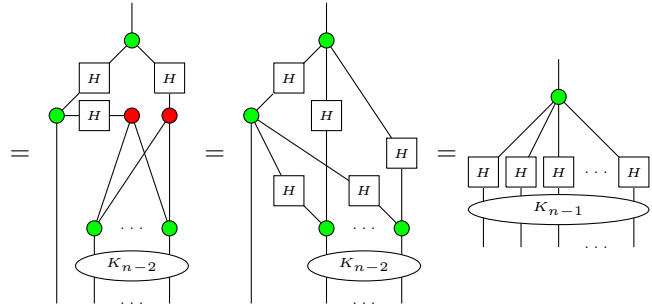
\begin{lemma}\label{LCStar} ~\vspace{-0.7cm}
\[ 
\vcenter{\hbox{\begin{tikzpicture}[xscale=0.50,yscale=0.50]
\useasboundingbox (-0.5,-0.5) rectangle (0.5,2.5);
\draw (0.00,2.00) -- (0.00,1.00) -- (0.00,0.00);
\filldraw[fill=white] (-0.40, 0.60)  -- (-0.40,1.40) -- (0.40,1.40) -- (0.40,0.60) -- (-0.40,0.60);
\draw (0.00,1.00) node{\tiny $H$};
\end{tikzpicture}
}} 
= ~
\vcenter{\hbox{\input{diag/axiomEuler-r.tex}}} 
~~\implies~~
\vcenter{\hbox{\input{diag/kn-l.tex}}} 
= ~
\vcenter{\hbox{\input{diag/kn-r.tex}}} 
\] 
\end{lemma}
\begin{proof} 
 \[
 \vcenter{\hbox{\input{diag/kn-1.tex}}} =
 \vcenter{\hbox{\input{diag/kn-2.tex}}} =
 \vcenter{\hbox{\input{diag/kn-3.tex}}} =
 \vcenter{\hbox{\input{diag/kn-4.tex}}} =
 \vcenter{\hbox{\input{diag/kn-5.tex}}}
 \]
 \[
 = \vcenter{\hbox{\input{diag/kn-6.tex}}} =
 \vcenter{\hbox{\input{diag/kn-7.tex}}} =
 \vcenter{\hbox{\input{diag/kn-r.tex}}}
 \]
\end{proof}

\subsubsection{General case}

The general case can be reduced to the previous case: first green
rotations can always be pushed through green dots for obtaining the
lhs of equation in Lemma \ref{LCStar}. After the application of the
lemma, one may have pairs of vertices having two edges (one coming
from the original graph, and the other from the complete graph). The
Hopf law is then used for removing these two edges.

\subsection{ Local Complementation Implies Euler Decomposition}
\label{sec:localcomp-impl-eulur}

\begin{lemma}\label{lem:decomp}
Local complementation implies the $H$-decomposition: ~Ê\vspace{-0.5cm}
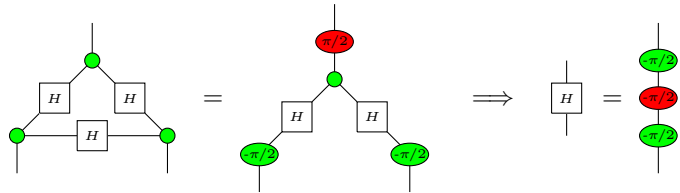
\[ 
\vcenter{\hbox{\input{diag/tri-l.tex}}} 
= ~
\vcenter{\hbox{\input{diag/tri-r.tex}}}
~~\implies~~
\vcenter{\hbox{\begin{tikzpicture}[xscale=0.50,yscale=0.50]
\useasboundingbox (-0.5,-0.5) rectangle (0.5,2.5);
\draw (0.00,2.00) -- (0.00,1.00) -- (0.00,0.00);
\filldraw[fill=white] (-0.40, 0.60)  -- (-0.40,1.40) -- (0.40,1.40) -- (0.40,0.60) -- (-0.40,0.60);
\draw (0.00,1.00) node{\tiny $H$};
\end{tikzpicture}
}} 
= ~
\vcenter{\hbox{\input{diag/axiomEuler-r.tex}}} 
\] 
\end{lemma}
\begin{proof}
The local complementation property can be rewritten as follows:
 \[
 \vcenter{\hbox{\input{diag/Euler1.tex}}} = ~
 \vcenter{\hbox{\input{diag/Euler2.tex}}}
 \]
 then
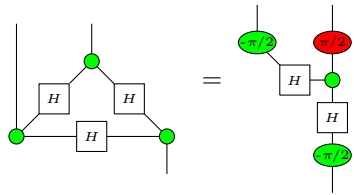
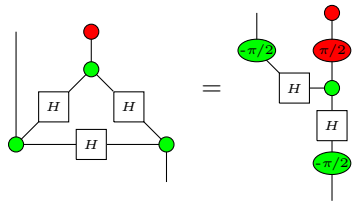
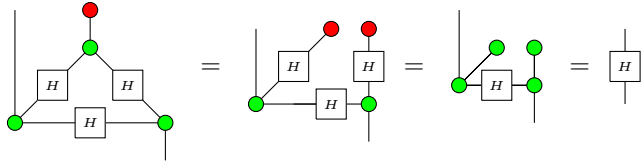
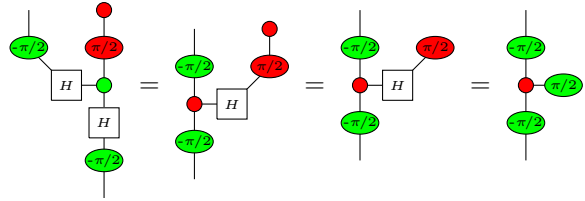
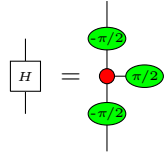
 \[
 \vcenter{\hbox{\input{diag/Euler3.tex}}} = ~
 \vcenter{\hbox{\input{diag/Euler4.tex}}}
 \]
 Since,
   \[
 \vcenter{\hbox{\input{diag/Euler3.tex}}} = ~
 \vcenter{\hbox{\input{diag/Euler5.tex}}} =~
 \vcenter{\hbox{\input{diag/Euler6.tex}}} =~
 \vcenter{\hbox{\begin{tikzpicture}[xscale=0.50,yscale=0.50]
\useasboundingbox (-0.5,-0.5) rectangle (0.5,2.5);
\draw (0.00,2.00) -- (0.00,1.00) -- (0.00,0.00);
\filldraw[fill=white] (-0.40, 0.60)  -- (-0.40,1.40) -- (0.40,1.40) -- (0.40,0.60) -- (-0.40,0.60);
\draw (0.00,1.00) node{\tiny $H$};
\end{tikzpicture}
}}
 \]
 And
     \[
 \vcenter{\hbox{\input{diag/Euler4.tex}}} =~
   \vcenter{\hbox{\input{diag/Euler7.tex}}} =~
 \vcenter{\hbox{\input{diag/Euler8.tex}}} =~
 \vcenter{\hbox{\input{diag/tri_7.tex}}}
 \]
 So 
 \[
 \vcenter{\hbox{\begin{tikzpicture}[xscale=0.50,yscale=0.50]
\useasboundingbox (-0.5,-0.5) rectangle (0.5,2.5);
\draw (0.00,2.00) -- (0.00,1.00) -- (0.00,0.00);
\filldraw[fill=white] (-0.40, 0.60)  -- (-0.40,1.40) -- (0.40,1.40) -- (0.40,0.60) -- (-0.40,0.60);
\draw (0.00,1.00) node{\tiny $H$};
\end{tikzpicture}
}} =
 \vcenter{\hbox{\input{diag/tri_7.tex}}}
 \]
 Complementing the above equation with $H$ on both sides, we obtain:
 \[
 \vcenter
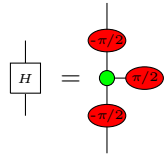
{\hbox{\begin{tikzpicture}[xscale=0.50,yscale=0.50]
\useasboundingbox (-0.5,-0.5) rectangle (0.5,2.5);
\draw (0.00,2.00) -- (0.00,1.00) -- (0.00,0.00);
\filldraw[fill=white] (-0.40, 0.60)  -- (-0.40,1.40) -- (0.40,1.40) -- (0.40,0.60) -- (-0.40,0.60);
\draw (0.00,1.00) node{\tiny $H$};
\end{tikzpicture}
}} =
 \vcenter{\hbox{\input{diag/tri_8.tex}}}
 \]
 Finally,
 \[
 \vcenter
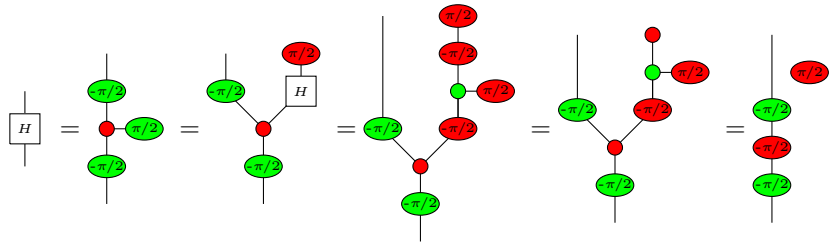
{\hbox{\begin{tikzpicture}[xscale=0.50,yscale=0.50]
\useasboundingbox (-0.5,-0.5) rectangle (0.5,2.5);
\draw (0.00,2.00) -- (0.00,1.00) -- (0.00,0.00);
\filldraw[fill=white] (-0.40, 0.60)  -- (-0.40,1.40) -- (0.40,1.40) -- (0.40,0.60) -- (-0.40,0.60);
\draw (0.00,1.00) node{\tiny $H$};
\end{tikzpicture}
}} =
 \vcenter{\hbox{\input{diag/tri_7.tex}}} =
 \vcenter{\hbox{\input{diag/tri_9.tex}}} =
 \vcenter{\hbox{\input{diag/tri_10.tex}}} =
 \vcenter{\hbox{\input{diag/tri_11.tex}}} =
 \vcenter{\hbox{\input{diag/tri_12.tex}}}
 \]
 which is the desired decomposition.
\end{proof}

\noindent This completes the proof of Theorem~\ref{thm:LC}.  Note that
we have shown the equivalence of two equations, both of which were
expressible in the graphical language.  What remains to be established
is that these properties---and here we focus on the decomposition of
$H$---are not derivable from the axioms already in the  system. To do
so we define a new interpretation functor.

Let $n \in \mathbb{N}$ and define $\denote{\cdot}_n$ exactly as
$\denote{\cdot}$ with the following change:
\[
\denote{p_X(\alpha)}_n = \denote{p_X(n\alpha)}
\qquad\qquad
\denote{p_Z(\alpha)}_n = \denote{p_Z(n\alpha)}
\]
Note that $\denote{\cdot} = \denote{\cdot}_1$.  Indeed, for all 
$n$, this functor preserves all the axioms introduced in 
Section~\ref{sec:graphical-formalism}, so its image is indeed a valid
model of the  theory.  However we have the following inequality
\[
  \denote{H}_n \neq \denote{p_Z(-\pi/2)}_n\circ 
                     \denote{p_X(-\pi/2)}_n\circ 
                     \denote{p_Z(-\pi/2)}_n
\]
for example, in $n=2$, hence the Euler decomposition is not derivable
from the axioms of the theory.

\section{Conclusions}
\label{sec:conclusions}

We studied graph states in an abstract axiomatic setting and saw that
we could prove Van den Nest's theorem if we added an additional
axiom to the theory.  Moreover, we proved that the
$\pi/2$-decomposition of $H$ is exactly the extra power which is
required to prove the theorem, since we prove that the
Van den Nest theorem is true if and only if $H$ has a $\pi/2$
decomposition.  It is worth noting that the system without $H$ is
already universal in the sense every unitary map is expressible, via
an Euler decomposition.  The original system this contained two
representations of $H$ which
could not be proved equal; it's striking that removing this ugly
wart on the theory turns out to necessary to prove a non-trivial
theorem.   In closing we note that this seemingly abstract high-level
result was discovered by studying rather concrete 
problems of measurement-based quantum computation.

\bibliography{CIE}

\end{document}